\documentclass[journal]{IEEEtran}
\usepackage{graphics} 
\usepackage{epsfig} 
\usepackage{mathptmx} 
\usepackage{times} 
\usepackage{multicol}
\usepackage{floatrow}
\usepackage{multirow}
\usepackage{amsmath}
\usepackage{amssymb}
\usepackage{amsthm}
\usepackage{graphicx}
\usepackage{color}
\usepackage{stfloats}
\usepackage{cite}

\newtheorem{theorem}{\indent Theorem}
\newtheorem{lemma}[theorem]{\indent Lemma}
\newtheorem{corollary}[theorem]{\indent Corollary}
\newtheorem{proposition}[theorem]{\indent Proposition}
\newtheorem{definition}[theorem]{\indent Definition}

\newtheorem{remark}[theorem]{\indent Remark}

\begin{document}

\title{Fault-tolerant Coherent $H^\infty$ Control for Linear Quantum Systems}

\author{Yanan~Liu,
        Daoyi~Dong,
        Ian R.~Petersen,
        Qing~Gao,
        \\ Steven X.~Ding,
        Shota~Yokoyama,
        and~Hidehiro~Yonezawa
\thanks{This work was supported by the Australian Research Council's Discovery Projects Funding Scheme under Project DP190101566 and Project DP180101805, the Air Force Office of Scientific Research under Agreement FA2386-16-1-4065, the Centres of Excellence under Grant CE170100012, the Alexander von Humboldt Foundation of Germany, and the U. S. Office of Naval Research Global under Grant N62909-19-1-2129.}
\thanks{Y. Liu, S. Yokoyama and H. Yonezawa are with the School of Engineering and Information Technology, University of New South Wales, Canberra, ACT 2600, Australia, and with Center for Quantum Computation and Communication Technology, Australian Research Council, Canberra, ACT 2600, Australia
        {(\tt\small yaananliu@gmail.com; s.yokoyama@adfa.edu.au; h.yonezawa@adfa.edu.au).}}
\thanks{D. Dong is with the School of Engineering and Information Technology, University of New South Wales, Canberra, ACT 2600, Australia, and with Institute for Automatic Control and Complex Systems (AKS), University of Duisburg-Essen, 47057 Duisburg, Germany {(\tt\small daoyidong@gmail.com).}}
\thanks{I. R. Petersen is with the Research School of Electrical, Energy and Materials Engineering, The Australian National University, Canberra, ACT 2601, Australia
        {( \tt\small i.r.petersen@gmail.com).}}
\thanks{Q. Gao is with the School of Automation Science and Electrical Engineering, State Key Laboratory of Software Development Environment, Beijing, and also with Advanced Innovation Center for Big Data and Brain Computing, Beihang University, Beijing 100191, China {(\tt\small qing.gao.chance@gmail.com).}}
\thanks{S. X. Ding is with the Institute for Automatic Control and Complex Systems (AKS), University of Duisburg-Essen, 47057 Duisburg, Germany {(\tt\small  steven.ding@uni-due.de).}}}

%



\maketitle

\begin{abstract}
Robustness and reliability are two key requirements for developing practical quantum control systems. The purpose of this paper is to design a coherent feedback controller for a class of linear quantum systems suffering from Markovian jumping faults so that the closed-loop quantum system has both fault tolerance and $H^\infty$ disturbance attenuation performance. This paper first extends the physical realization conditions from the time-invariant case to the time-varying case for linear stochastic quantum systems. By relating the fault tolerant $H^\infty$ control problem to the dissipation properties and the solutions of Riccati differential equations, an $H^\infty$ controller for the quantum system is then designed by solving a set of linear matrix inequalities (LMIs). In particular, an algorithm is employed to introduce additional noises and to construct the corresponding input matrices to ensure the physical realizability of the quantum controller. For real applications of the developed fault-tolerant control strategy, we present a linear quantum system example from quantum optics, where the amplitude of the pumping field randomly jumps among different values. It is demonstrated that a quantum $H^\infty$ controller can be designed and implemented using some basic optical components to achieve the desired control goal.\end{abstract}


%
\IEEEpeerreviewmaketitle

\section{Introduction}
%
%
%
%
\IEEEPARstart{D}{eveloping} robust and reliable quantum control systems is a fundamental task with practical significance in implementing various quantum technologies \cite{xiang2017coherent,xiang2016performance, wu2019learning, li2009ensemble, guo2018optimal, dong2019learning, kosut2013robust, ge2019robust, wang2018representation, Guo2019vanishing, altafini2012modeling, shu2020attosecond}. In practice, quantum systems may suffer from various kinds of faults. For example, the fluctuations of the lasers in quantum optics or fault operations on the generators of quantum resources may introduce fault signals, thereby deteriorating the performance of the system or causing the system to be unstable \cite{wang2016fault}. However, many unique features of quantum systems, such as measurement reduction and noncommutative observables \cite{dong2010quantum}, make fault-tolerant control strategies for classical systems difficult to be extended to their quantum counterparts \cite{wang2016fault, ding2008model, blanke2006diagnosis}. It is one of the main goals of this paper to develop a fault-tolerant feedback control approach for a class of linear quantum systems with fault signals. Feedback control, including measurement-based feedback control and coherent feedback control, plays an important role in quantum control theory since it has the capacity to suppress uncertainties and noises, thus having good robustness \cite{liu2016lyapunov, liu2019filter}. In measurement-based feedback control scheme, a measurement device is used to extract the plant information which is then fed back to the original system through the control channel to achieve desired closed-loop behavior \cite{belavkin1992quantum2, wiseman1994quantum}. Quantum measurement is different from the classical one, in the sense that it inevitably causes quantum state collapsion and introduces additional stochastic noises \cite{van2005feedback}. In addition, the time to process the measurement outcomes and calculate the control signal cannot be ignored in general, which causes a time delay problem in measurement-based feedback control \cite{liu2019feedback}. Coherent feedback control, where the controller itself is also a quantum system and has no such disadvantages, is used in our approach \cite{R1, nurdin2009coherent, maalouf2010coherent, zhang2010direct, yamamoto2014coherent}.

Unlike the cases in classical control systems, where we usually assume that all the designed controllers are physically realizable \cite{R1}, the controller designed using mathematical models for a quantum plant may not correpond to a real physical quantum system. Hence, the physical realization in designing a coherent feedback controller needs to be considered. James \emph{et al.} \cite{R1} have deduced necessary and sufficient conditions of physical realizability for a class of quantum systems described by quantum stochastic differential equations (QSDEs). They also pointed out that if the designed quantum controller does not satisfy the realizable conditions, one can introduce additional quantum noises and adjust the corresponding input matrices to make the controller physically realizable. For a class of linear quantum systems described by complex transfer function matrices, a construction algorithm has been proposed to physically synthesize them in \cite{petersen2011cascade}. Nurdin \emph{et al.} \cite{nurdin2009network, nurdin2010synthesis, nurdin2010synthesis2} proposed several ways to synthesize a quantum system described by time-invariant linear differential equations, and discussed how to implement the quantum systems using some basic optical components. These synthesis schemes are of significance in designing coherent feedback control strategies even though the experimental setup is often complex and challenging.

On the other hand, $H^\infty$ control is a well known robust control method that has been used in both classical and quantum systems \cite{petersen1991first,zhou1988algebraic,ravi1991h,R1}. A quantum version of the standard dissipation properties has been proposed in \cite{R1}, where the $H^\infty$ control problem for quantum systems was formulated using two Riccati equations. By solving these Riccati equations, a controller is obtained and can be implemented as either a fully classical system, a purely quantum system or a mixture of both quantum and classical elements. While \cite{R1} only considered the cases of time-invariant quantum systems, in practical applications, time-varying linear quantum systems are often encountered. A dynamic game approach to designing a classical $H^\infty$ controller for a class of time-varying linear quantum systems has been proposed in \cite{maalouf2012time}, by recognising the equivalence between a quantum system and a corresponding auxiliary classical stochastic system. A linear quadratic Gaussian (LQG) optimal controller has been designed in \cite{hassen2010lqg} to optimize the squeezing level achieved in one of the quadratures of the fundamental optical field for the time-varying quantum systems.

This paper aims to solve the time-varying $H^\infty$ coherent feedback control problem for a linear quantum system suffering from a fault signal. The dissipation properties of the time-varying quantum systems are first presented, by which the $H^\infty$ control problem is formulated in several Riccati differential equations and a group of LMIs. The fault under consideration is modelled as a Markovian chain on a probability space \cite{gao2016fault, gao2019design}. The physical realization conditions for time-varying quantum systems are then presented and an algorithm is given to construct a physically meaningful quantum controller. In many practical applications, quantum optical systems have shown powerful potential for developing future quantum technology \cite{yonezawa2012quantum, asavanant2019generation, hansch1999laser, serikawa2016creation}. In this paper, we use a squeezer that has been widely used in quantum optics \cite{serikawa2016creation} to test the effectiveness of our control approach. A purely quantum $H^\infty$ controller consisting of basic optical components is designed to ensure that the system has desired robust performance even when suffering from faults.

The contributions of this paper are summarized as follows.
\begin{enumerate}
\item[$\bullet$] The dissipation properties of time-varying quantum systems are illustrated and used to design an $H^\infty$ quantum controller.
\item[$\bullet$] The physical realization conditions for time-varying linear quantum systems are investigated.
\item[$\bullet$] The proposed quantum fault-tolerant $H^\infty$ control method is applied to quantum optical systems, where an optical parametric amplifier is recognized to be a time-varying linear quantum system with a fault signal.
\item[$\bullet$] A quantum controller is designed using several basic optical components to achieve fault-tolerant coherent $H^\infty$ control for a class of quantum systems.
\end{enumerate}

The rest of this paper is organized as follows. Section \ref{secmodel} presents the system model and the problem formulation. In Section \ref{dissipation}, a main theorem is obtained to illustrate the equivalence between the dissipation properties, $H^\infty$ control problem and relevant Riccati differential equations. Section \ref{controldesign} presents a controller design in terms of LMIs. In Section \ref{application}, we provide an application of fault-tolerant quantum control for a class of quantum systems in quantum optics. A squeezer where the pumping field suffers from a fault signal is considered. A purely quantum controller is implemented by some basic optical components. Section \ref{conclusion} concludes this paper.

\textit{Notation}: $A$ represents an operator in Hilbert space, and is a matrix with proper dimension. $A^T$ represents the transpose of $A$; $A^\dagger$ is adjoint of $A$, and $A^{-1}$ is the inverse of $A$; $\parallel A\parallel_{\infty}$ represents the $H^\infty$ norm of the operator $A$. ${\rm Tr}(A)$ is the trace of $A$; $\Im(A)$ represents the imaginary part of $A$; $\mbox{i}$ means imaginary unit, i.e., $\mbox{i} = \sqrt{-1}$; $\mbox{I}$ is the identity matrix with proper dimension; $\hbar$ is the reduced Planck constant; $\mathbb{S}$ represents the state space.

\section{System Model and Problem Formulation}
\label{secmodel}
Linear quantum systems are usually described by time-invariant linear differential equations. However, when the system suffers from a fault process, the equations may no longer be time invariant. In this paper, we consider the following time-varying linear quantum system:
\begin{equation}
\label{eq8}
\begin{aligned}
dx(t)&=A(t)x(t)dt+Bd\omega(t); x(0)=x_0,\\
dy(t)&=Cx(t)dt+Dd\omega(t),
\end{aligned}
\end{equation}
where $A(t)\in \mathbb{R}^{n\times n}, B\in \mathbb{R}^{n\times n_{\omega}}, C\in \mathbb{R}^{n_y \times n}, D \in \mathbb{R}^{n_y \times n_{\omega}}$, and $x(t)=[x_1(t), x_2(t), \cdots, x_n(t)]^T$ is a vector of self-adjoint possibly noncommutative system variables. The initial system variables satisfy the commutation relation \cite{R1}
\begin{equation}
\label{eq2}
[x_i(0),x_j(0)]=2\mbox{i}\Theta_{jk}, j,k=1, \cdots, n.
\end{equation}
Here the commutator is defined as $[A, B]=AB-BA$. $\Theta_{jk}$ is defined to be of one of the following form \cite{R1}
\begin{itemize}
\item Canonical if $\Theta={\rm diag}(J,\cdots,J)$;
\item Degenerate canonical if $\Theta={\rm diag}(0_{n'\times n'},J,\cdots,J)$,where $0<n'\leq n$.
\end{itemize}
Here, $J$ is the real skew-symmetric matrix $J=\begin{bmatrix}
0 & 1\\
-1 & 0
\end{bmatrix}$.
$\omega$ represents the disturbance input, and is assumed to have the form
\begin{equation}
\label{eq3}
d\omega(t)=\beta_{\omega}(t)dt+d\tilde{\omega}(t),
\end{equation}
where $\beta_{\omega}(t)$ is a self adjoint process and $\tilde{\omega}(t)$ is the noise part. The quantum noise satisfies It$\hat{o}$ table $d\tilde{\omega}(t)d\tilde{\omega}^T(t)=F_{\tilde{\omega}}dt$ with a nonnegative matrix $F_{\tilde{\omega}}$ \cite{belavkin1992quantum}. We write
\begin{equation}
\label{eq3-1}
S_{\tilde{\omega}}=\frac{1}{2}(F_{\tilde{\omega}}+F_{\tilde{\omega}}^T),
\end{equation} and $$T_{\tilde{\omega}}=\frac{1}{2}(F_{\tilde{\omega}}-F_{\tilde{\omega}}^T),$$ where $T_{\tilde{\omega}}$ satisfies the following equation
\begin{equation}
\label{eq4}
[d\tilde{\omega}(t), d\tilde{\omega}^T(t)]=d\tilde{\omega}(t)d\tilde{\omega}^T(t)-(d\tilde{\omega}(t)d\tilde{\omega}^T(t))^T=2T_{\tilde{\omega}}dt.
\end{equation}

The system \eqref{eq8} can be related to optical parametric amplifiers, where the laser field may often be treated in a classical way. If the laser device of the pumping field is subject to a fault process, a time-varying Hamiltonian will be introduced to the linear differential equations. In this case, the system Hamiltonian depends on the fault process and can be described as $H(F(t))$ where $F(t)$ is a fault process\cite{gao2016fault}. The following definition describes a time-varying open quantum harmonic oscillator.
\begin{definition}
The system \eqref{eq8} (with $\beta_{\omega}=0$) is said to be an open quantum harmonic oscillator if $\Theta$ is canonical and there exist a quadratic Hamiltonian $H=\frac{1}{2}x(0)^T R(t) x(0)$, with a real and symmetric Hamiltonian matrix $R(t)$ of dimension $n\times n$, and a coupling operator $L=\Lambda x(0)$, with complex-valued coupling matrix $\Lambda$ of dimension $n_{\omega}\times n$, such that
\[
\begin{aligned}
x_k(t)&=U(t)^{*}x_k(0)U(t), k=1, \cdots, n;\\
y_l(t)&=U(t)^{*}\omega_l(0)U(t), l=1, \cdots,n_y,
\end{aligned}
\]
where $\{U(t); t\geq0\}$ is an adapted process of unitary operators satisfying the following QSDE:
\[
\begin{aligned}
dU(t)&=\left(-\mbox{i}H(F(t))dt-\frac{1}{2}L^{\dagger}Ldt+[-L^{\dagger} L^T]\Gamma d\omega(t) \right)U(t),\\
U(0)&=\mbox{I}.
\end{aligned}
\]
In this case, the matrices $A, B, C, D$ are given by
\begin{equation}
\label{eq5}
A=2\Theta(R(F(t))+\Im(\Lambda^{\dagger}\Lambda)),
\end{equation}
\begin{equation}
\label{eq6}
B=2\mbox{i}\Theta[-\Lambda^{\dagger} \Lambda^T]\Gamma,
\end{equation}
\begin{equation}
\label{eq7}
C=P_{N_y}^T\begin{bmatrix}
\sum_{N_y} & 0_{N_y \times N_\omega}\\
0_{N_y \times N_\omega} & \sum_{N_y}
\end{bmatrix} \begin{bmatrix}
\Lambda+\Lambda^{\#}\\
-\mbox{i}\Lambda+\mbox{i}\Lambda^{\#}
\end{bmatrix}.
\end{equation}
Here, $N_{\omega}=\frac{n_{\omega}}{2}, N_y=\frac{n_y}{2}$. $\Gamma=P_{N_\omega}{\rm diag}_{N_\omega}(M)$ with $M=\frac{1}{2}\begin{bmatrix}
1 & \mbox{i}\\
1 &-\mbox{i}
\end{bmatrix}$, and $P_k$ is the permutation matrix satisfying $$P_k a=\begin{bmatrix}
a_1 &a_3 & \cdots & a_{2m-1} & a_2 &a_4 &\cdots&a_{2m}
\end{bmatrix}^T,$$
for an arbitrary vector $a=\begin{bmatrix}
a_1 & a_2 & \cdots &a_{2m}
\end{bmatrix}^T$. And $\sum_{N_y}=\begin{bmatrix}\mbox{I}_{N_y\times N_y} & 0_{N_y \times (N_\omega-N_y)} \end{bmatrix}$.
\end{definition}

The main goals of this paper are to analyze the dissipation properties of  the time-varying stochastic linear quantum systems described by \eqref{eq8}, and to design a coherent quantum feedback controller to achieve closed-loop $H^\infty$ performance.

\section{Dissipation Properties}
\label{dissipation}
In this section, we consider some dissipation properties for the time-varying quantum system \eqref{eq8}. Dissipation properties state the relation between storage function and the supply functions in terms of system energy \cite{willems1972dissipative}, and relevant discussions in the context of time-invariant linear quantum systems have been given in \cite{R1}. Here we make a further extension to the time-varying cases. In addition, we formulate the strict bounded real lemma, which will be used in the controller synthesis later.

We consider a quantum system described as follows:
\begin{equation}
\label{eq13}
\begin{aligned}
dx(t)&=A(t)x(t)dt+\begin{bmatrix}B&G\end{bmatrix}\begin{bmatrix}d\omega(t)^T & d\nu(t)^T\end{bmatrix}^T,\\
dz(t)&=Cx(t)dt+\begin{bmatrix}D & H\end{bmatrix}\begin{bmatrix}d\omega(t)^T & d\nu(t)^T\end{bmatrix}^T.
\end{aligned}
\end{equation}
Here, $d\omega(t)=\beta_{\omega}(t)dt+d\tilde{\omega}$ represents the disturbance input with the quantum noise $d\tilde{\omega}$, and $d\nu$ represents additional noise.

We first define a storage function $V(x(t))=x(t)^TP(t)x(t)$, where $P(t)$ is a time-varying positive definite symmetric matrix, and then define the following operator valued quadratic function $$\gamma(x, \beta_{\omega})=\begin{bmatrix}
x^T & \beta_{\omega}^T
\end{bmatrix}^T R \begin{bmatrix}
x\\
\beta_{\omega}
\end{bmatrix},$$
as the supply function, where $R$ is a constant real symmetric matrix.

\begin{definition}\cite{R1}
\label{defdissipation}
The system \eqref{eq13} is said to be dissipative with supply rate $\gamma(x, \beta_{\omega})$ if there exists a positive time-varying storage function $V(x(t))=x(t)^TP(t)x(t)$ and a constant $\lambda>0$ such that
\begin{equation}
\label{eq14}
\langle V(x(t)) \rangle+\int_0^t \langle \gamma(x(s), \beta_{\omega}(s)) \rangle ds\leq \langle V(x(0)) \rangle +\lambda t, \forall t>0,
\end{equation}
where $\langle V(x(t)) \rangle$ represents the expectation of the operator $V(x(t))$.
\end{definition}

The system \eqref{eq13} is said to be strictly dissipative if there exists a constant $\epsilon>0$ such that inequality \eqref{eq14} holds for the supply function with $R+\epsilon \mbox{I}$.

\begin{definition}\cite{R1}
\label{defboundedreal}
The quantum systems \eqref{eq13} is said to be bounded real with disturbance attenuation $g$ if the system is dissipative with
\begin{equation}
\label{eq15}
\begin{aligned}
\gamma(x, \beta_{\omega})&=\beta_z^T \beta_z-g^2 \beta_{\omega}^T\beta_\omega\\
&=\begin{bmatrix}
x^T & \beta_{\omega}^T
\end{bmatrix} \begin{bmatrix}
C^TC & C^TD\\
D^TC & D^TD-g^2\mbox{I}
\end{bmatrix}\begin{bmatrix}
x\\
\beta_{\omega}
\end{bmatrix},
\end{aligned}
\end{equation}
where $\beta_z(t)=Cx(t)+D\beta_{\omega}(t)$. Also we say that the system \eqref{eq13} is strictly bounded real with disturbance attenuation $g$ if the system is strictly dissipative with supply rate \eqref{eq15}.
\end{definition}

With these definitions, the following theorem states the relationship between the dissipation properties and the Riccati differential equations, as well as the $H^\infty$ control problem, which will be used to design a coherent controller.
\begin{theorem}
\label{the1}
For the system \eqref{eq13}, the following four statements are equivalent:
\begin{itemize}
\item[1)]The system \eqref{eq13} is strictly bounded real with disturbance attenuation $g$;
\item[2)]There exists a positive definite matrix $\tilde{P}(t)$ such that
\begin{align*}
&\dot{\tilde{P}}(t)+A(t)^T\tilde{P}(t)+\tilde{P}(t)A(t)+C^T C\\
&+(C^T D+\tilde{P}(t)B)(g^2 \mbox{I}-D^TD)^{-1}(D^TC+B^T \tilde{P}(t))\\
&<0;
\end{align*}
\item[3)]The Riccati differential equation
\begin{align*}
&\dot{P}(t)+A(t)^TP(t)+P(t)A(t)+C^TC\\
&+(C^TD+P(t)B)(g^2\mbox{I}-D^TD)^{-1}(D^TC+B^TP(t))\\
&=0
\end{align*}
has a stabilizing solution $P(t)\geq 0$; and
\item[4)]The homogeneous system $\dot{x}(t)=A(t)x(t)$ is exponentially stable, and the operator mapping $\omega$ to $z$ satisfies $\parallel T_{z\omega}\parallel_{\infty}<g$.
\end{itemize}
\end{theorem}
\begin{proof}

Since the equivalence between statements 2), 3) and 4) has been proved in \cite{chen2000strict}, we here only prove the equivalence between 1) and 2).

For a given storage function $V(x(t))=x(t)^T P(t) x(t)$, we calculate
\begin{equation}
\begin{aligned}
\label{eq16}
d\langle V(x(t)) \rangle &=\left\langle dx^T(t)\cdot \left[P(t)x(t)\right]+x^T(t)\cdot d\left[P(t)x(t)\right] \right\rangle \\
&=\langle x^T(t)\left( A^T(t)P(t)+P(t)A(t)+\dot{P}(t) \right)x(t)\\
&+\beta_{\omega}^T(t)B^TP(t)x(t)+x^T(t)P(t)B\beta_{\omega}(t)+\lambda_0 \rangle,
\end{aligned}
\end{equation}
where $\dot{P}(t)=\frac{dP(t)}{dt}$ and
$$\lambda_0={\rm Tr}\left\{ \begin{bmatrix} B^T\\G^T \end{bmatrix}X\begin{bmatrix}B & G \end{bmatrix}F \right\},$$
and $F$ is defined as
$$Fdt=\begin{bmatrix}d\omega\\d\nu\end{bmatrix}\begin{bmatrix} d\omega^T & d\nu^T\end{bmatrix}.$$

Suppose $\rho$ is an initial Gaussian state, and let $E_0$ denote the expectation with respect to a random state $\phi$. We have $\langle V(x(t))\rangle=\langle \rho, E_0\left[ V(x(t))\right]\rangle$ \cite{parthasarathy2012introduction}. If the system \eqref{eq13} is bounded real with disturbance attenuation $g$, then we have
\begin{equation}
\label{eq17}
\begin{aligned}
&\Bigg \langle \rho, \int_0^t E_0 \Big[ x^T(s)\left( A^T(t)P(t)+P(t)A(t)+\dot{P}(t) \right)x(s)\\
&+\beta_{\omega}^T(s)B^TP(t)x(s)+x^T(s)P(t)B\beta_{\omega}(s)\\
&+\lambda_0 +\gamma(x(s),\beta_{\omega}(s))\Big]ds \Bigg\rangle \leq \lambda t,
\end{aligned}
\end{equation}
where $\gamma(x(s),\beta_{\omega}(s))$ is defined as \eqref{eq15}.

When $t\rightarrow 0$, we obtain
\begin{equation}
\label{eq18}
\begin{aligned}
& \Big \langle \rho, x^T(A^TP(t)+P(t)A)x+\beta_{\omega}^TB^TP(t)+x^TP(t)B\beta_{\omega} \\
&+\lambda_0
+\begin{bmatrix}
x^T & \beta_{\omega}^T
\end{bmatrix} \begin{bmatrix}
C^TC & C^TD \\
D^TC & D^TD-g^2 \mbox{I}
\end{bmatrix} \begin{bmatrix}
x\\
\beta_{\omega}
\end{bmatrix} \Big\rangle \leq \lambda.
\end{aligned}
\end{equation}
According to Lemma \ref{lemma1} in Appendix B, we have
\begin{equation}
\label{eq19}
\begin{bmatrix}
\dot{P}(t)+A^T P(t)+P(t)A+C^TC & C^TD+P(t)B\\
B^TP(t)+D^TC & D^TD-g^2\mbox{I}
\end{bmatrix}\leq 0.
\end{equation}

Furthermore, the system is strictly bounded real with disturbance attenuation $g$ if and only if there exists a real positive definite symmetric matrix such that the following matrix inequality is satisfied:
\begin{equation}
\begin{bmatrix}
\dot{P}(t)+A^T P(t)+P(t)A+C^TC & C^TD+P(t)B\\
B^TP(t)+D^TC & D^TD-g^2\mbox{I}
\end{bmatrix}< 0,
\end{equation}
which means 2).

From statement 2) to statement 1):

Based on Schur's complement, statement 2) implies \eqref{eq19}. Then by choosing the storage function as $V(x)=x^T P(t)x$ and by setting $\lambda_0=B^TP(t)B F$, one obtains \eqref{eq18} from \eqref{eq19} and moreover,
\begin{equation}
\label{eq20}
\langle V(x(t)) \rangle-\langle V(x(0)) \rangle+\int_0^t \langle \gamma(x(s),\beta_{\omega}(s)) \rangle ds\leq \lambda_0 t .
\end{equation}
From Definition 2 and Definition 3, one has that the system \eqref{eq13} is strictly bounded real with disturbance attenuation $g$.

The proof is then completed.
\end{proof}


\section{Coherent $H^\infty$ control design}
\label{controldesign}
Quantum optical systems usually are sensitive to external disturbance. In this paper, we consider a linear quantum system suffering from abrupt variation in its parameters, structure or system dynamics such that the system dynamics may randomly transit between a finite number of different modes, named faulty modes. It is then appropriate to model the fault process on a probability space $(\Omega, \mathcal{F}, \mathcal{P})$ by a continuous-time Markov chain $\{F(t)\}_{t\geq 0}$ \cite{gao2016fault}, which results in a Markovian jump linear quantum system. To be specific, $F(t)$ values within a finite set $\mathbb{S}=\{e_1, e_2, \cdots, e_N\}$ for an integer $N$. The transition rate matrix is priorly known as $\Pi=(\pi_{jk})\in \mathbb{R}^{N\times N}$, with $\pi_{jj}=-\sum_{j\neq k}\pi_{jk}$, and $\pi_{jk}\geq 0, j\neq k$. In this section, we develop a coherent $H^\infty$ control design for a class of linear quantum systems whose Hamiltonian is dependent on the fault process $F(t)$.

\subsection{Closed-loop systems}
\label{closedsystem}
The system with a disturbance input and a control input is described as
\begin{equation}
\label{eq31}
\begin{aligned}
dx(t)&=A(F(t))x(t)dt+B_1d\omega(t)+B_2du(t),\\
dz(t)&=C_1x(t)dt+D_1 du(t),\\
dy(t)&=C_2x(t)dt+D_{2}d\omega(t),
\end{aligned}
\end{equation}
 $A(F(t))$ takes finite values in $(A_1, A_2, \cdots, A_N), A_i=A(e_i)$ since the fault process $F(t)$ has been assumed to be a Markov chain, which has values within a finite set $\mathbb{S}=\{e_1, e_2, \cdots, e_N\}$.

Assume that the controller is described by the following dynamical equations
\begin{equation}
\label{eq32}
\begin{aligned}
d\xi(t)&=\mathcal{A}(t)\xi(t)dt+\mathcal{B}(t)dy(t)+\mathcal{E}(t)d\nu_K(t),\\
du(t)&=\mathcal{C}(t)\xi(t)dt+\mathcal{D}(t)d\nu_K(t),
\end{aligned}
\end{equation}
where $\xi(t)=\begin{bmatrix}
\xi_1(t) & \xi_2(t)& \cdots & \xi_{n_k}
\end{bmatrix}^T$ is a vector of self-adjoint controller variables. The noise $\nu_K$ is a vector of noncommutative Wiener processes satisfying the It$\hat{o}$ table with canonical Hermitian It$\hat{o}$ matrix $F_{\nu_{K}}$. To coincide with a Markovian jump linear plant, the controller also is assumed to jump between different modes with $(\mathcal{A}_{1}, \mathcal{B}_{1}, \mathcal{C}_{1}), \cdots, \left(\mathcal{A}_{N}, \mathcal{B}_{N}, \mathcal{C}_{N}\right)$.

We obtain the closed-loop systems by identifying $\beta_u(t)=\mathcal{C}(t)\xi(t)$ and denoting $\eta=\begin{bmatrix}
x(t)& \xi(t)
\end{bmatrix}^T$ as
\begin{equation}
\label{eq35}
\begin{aligned}
d\eta(t)&=\begin{bmatrix}
A_i & B_2\mathcal{C}_i\\
\mathcal{B}_iC_2 & \mathcal{A}_i
\end{bmatrix}\eta(t)dt+\begin{bmatrix}
B_1\\
\mathcal{B}_iD_2
\end{bmatrix}d\omega(t)\\
&+\begin{bmatrix}
B_2 \mathcal{D}_i\\
\mathcal{E}_i
\end{bmatrix}d\nu_K(t),\\
dz(t)&=\begin{bmatrix}
C_1 & D_1 \mathcal{C}_i
\end{bmatrix}\eta(t)dt+ D_1 \mathcal{D}_id\nu_K(t).
\end{aligned}
\end{equation}

The control objective here is to design a controller \eqref{eq32} such that the closed-loop system \eqref{eq35} is strictly bounded real with a given disturbance attenuation $g$, that is, there exists a positive definite matrix $P(t)$ such that
\begin{equation}
\label{eq36}
\begin{aligned}
&\left\langle \eta^T(t)P(t)\eta(t) \right\rangle+\int_0^t \Big\langle \beta_z^T(s)\beta_z(s)-g^2 \beta_{\omega}^T(s) \beta_{\omega}(s)\\
&+\epsilon\eta^T(s)\eta(s)+\epsilon\beta_{\omega}^T(s)\beta_{\omega}(s) \Big\rangle ds\\
&\leq \langle \eta^T(0)P_0\eta(0) \rangle +\lambda t, \forall t>0.
\end{aligned}
\end{equation}

\subsection{$H^\infty$ controller design}
\label{controllerdesign}
Substituting $d\omega(t)=\beta_\omega (t)dt+d\tilde{\omega}(t)$ and $du(t)=\beta_u(t)dt+d\tilde{u}(t)$ into \eqref{eq35}, we have
\begin{equation}
\label{eq46}
\begin{aligned}
d\eta(t)&=\tilde{A}_i\eta(t)dt+\tilde{B}_{1i}\beta_{\omega}(t)dt+\tilde{B}_{1i}\tilde{\omega}(t)+\tilde{B}_{2i}\nu_K(t),\\
dz(t)&=\tilde{C}_i\eta(t)dt+\tilde{D}_id\nu_K(t).
\end{aligned}
\end{equation}
Here,
$$\tilde{A}_i=\begin{bmatrix}
A_i & B_2\mathcal{C}_i\\
\mathcal{B}_iC_2 & \mathcal{A}_i
\end{bmatrix},$$
$$\tilde{B}_{1i}=\begin{bmatrix}
B_1\\
\mathcal{B}_iD_2
\end{bmatrix}, \tilde{B}_{2i}=\begin{bmatrix}
B_2 \mathcal{D}_i\\
\mathcal{E}_i
\end{bmatrix},$$
$$\tilde{C}_i=\begin{bmatrix}
C_1 & D_1\mathcal{C}_i
\end{bmatrix},$$
$$\tilde{D}_i=D_1\mathcal{D}_i.$$

For the quantum system, we define $\langle V(\eta(t))\rangle=\langle \eta ^T (t)P(t) \eta (t)\rangle$ with a positive-definite matrix $P(t)$ and
\begin{equation}
\label{eq48_1}
\begin{aligned}
&\gamma(\eta(t),\beta_{\omega} (t))\\
&=\beta_z (t)^T \beta_z (t)-g^2 \beta_{\omega}^T (t) \beta_{\omega} (t)+\epsilon \eta^T(t)\eta(t)+\epsilon \beta_\omega^T(t)\beta_\omega(t)\\
&=\eta^T (t) [\tilde{C}_i^T \tilde{C}_i+\epsilon \mbox{I}]\eta(t)-(g^2-\epsilon)\mbox{I} \beta_{\omega}^T (t) \beta_{\omega} (t).
\end{aligned}
\end{equation}
Substituting \eqref{eq48_1} into \eqref{eq14}, we obtain
\begin{equation}
\label{eq49}
\begin{aligned}
&\langle\eta^T (t)P(t)\eta(t)+\int_0^t \left \langle \eta^T (s)\left[ \tilde{C}_i^T \tilde{C}_i+\epsilon \mbox{I}\right]\eta(s)\right \rangle ds\\
-&\int_0^t \left \langle(g^2-\epsilon) \beta_{\omega}^T (s) \beta_{\omega} (s)\right \rangle ds\\
&\leq \langle \eta_0^T P_0\eta_0 \rangle+\lambda t.
\end{aligned}
\end{equation}

Define $Q(t)=\frac{1}{2}\langle \eta(t)\eta^T(t)+(\eta(t)\eta^T(t))^T\rangle$, and we have $\langle \eta^T(t)P(t)\eta(t)\rangle={\rm Tr}[ \tilde{P}(t)Q(t)]$, where $$\tilde{P}(t)=\begin{bmatrix}
P(t) & 0\\
0 & 0
\end{bmatrix},$$
$$\left \langle \eta^T(t)\left[\tilde{C}_i^T\tilde{C}_i+\epsilon \mbox{I}\right]\eta(t) \right \rangle={\rm Tr}\left\{ \left[\tilde{C}_i^T\tilde{C}_i+\epsilon \mbox{I}\right]Q(t)\right\}.$$
Hence, \eqref{eq49} becomes
\begin{equation}
\label{eq50}
\begin{aligned}
&{\rm Tr}\langle \tilde{P}(t)Q(t)\rangle +\int_0^t {\rm Tr}\left\{ \left[\tilde{C}_i^T\tilde{C}_i+\epsilon \mbox{I}\right]Q(t)\right\}ds\\
&-(g^2-\epsilon) \int_0^t\langle \beta_{\omega}^T(s)\beta_{\omega}(s) \rangle ds \\
&\leq \langle \eta^T(0) P_0 \eta(0) \rangle +\lambda t.
\end{aligned}
\end{equation}

The derivative of $Q(t)$ is
\begin{equation}
\label{52_1}
\begin{aligned}
dQ(t)&=\tilde{A}_iQ(t)+Q(t)\tilde{A}_i^T\\
&+\langle \eta (t)\rangle \beta_{\omega}^T(t)\tilde{B}_{1i}^T dt+\tilde{B}_{1i}\beta_{\omega}(t)\langle \eta^T (t)\rangle dt\\
&+\tilde{B}_{1i} S_{\tilde{\omega}}(t)\tilde{B}_{1i}^Tdt+\tilde{B}_{2i} S_{\nu_K}(t)\tilde{B}_{2i}^Tdt.
\end{aligned}
\end{equation}

We consider a corresponding matrix $Q'(t)={\rm E}(\eta_c(t)\eta_c^T(t))$, where $E(\cdot)$ denotes the stochastic expectation, and $\eta_c(t)$ is the system variable of a classical system defined as
\begin{equation}
\label{eq47}
\begin{aligned}
d\eta_c (t)&=\tilde{A}_i \eta_c (t)dt+\tilde{B}_{1i}\beta_{\omega} (t)dt+\tilde{B}_{1i}S_{\tilde{\omega}}^{1/2}d\tilde{\omega}(t)\\
&+\tilde{B}_{2i} S_{\nu_K}^{1/2} d\nu_K (t),\\
dz_c(t)&=\tilde{C}_i\eta_c(t)dt+\tilde{D}_iS_{\nu_K}^{1/2}d\nu_K(t),
\end{aligned}
\end{equation}
where $\eta_c(t)=\left[
\begin{smallmatrix}
x_c(t)\\
\xi_c(t)
\end{smallmatrix}\right]$, and $\eta_c(0)=\eta_{c0}=\begin{bmatrix}
x_c(0)\\
\epsilon_c(0)
\end{bmatrix}$.

The system \eqref{eq47} can be taken as a closed-loop system composed of a plant
\begin{equation}
\label{eq43}
\begin{aligned}
dx_c(t)&=A(F(t))x_c(t)dt+B_1d\omega_c(t)+B_2du_c(t),\\
dz_c(t)&=C_1x_c(t)dt+D_1 du_c(t),\\
dy_c(t)&=C_2x_c(t)dt+D_{2}d\omega_c(t),
\end{aligned}
\end{equation}
and a controller
\begin{equation}
\label{eq44}
\begin{aligned}
d\xi_c(t)&=\mathcal{A}(t)\xi_c(t)dt+\mathcal{B}(t)dy_c(t)+\mathcal{E}(t)d\nu_c(t),\\
du_c(t)&=\mathcal{C}(t)\xi_c(t)dt+\mathcal{D}(t)d\nu_c(t).
\end{aligned}
\end{equation}
Here, $d\omega_c(t)=\beta_{\omega}(t)dt+S_{\tilde{\omega}}^{1/2}d\tilde{\omega}(t)$, $d\nu_c(t)=S_{\nu_K}^{1/2}d\nu_K(t)$, and $x_c(0)$ is a Gaussian random vector with mean $\hat{x}_{c0}$ and convariance matrix $Y_{c0}$. $S_{\tilde{\omega}}$ and $S_{\tilde{u}}$ are defined as in \eqref{eq3-1}.

The derivative of $Q'(t)$ is calculated as
\begin{equation}
\label{eq52}
\begin{aligned}
dQ'(t)&=\tilde{A}_iQ'(t)+Q'(t)\tilde{A}_i^T\\
&+{\rm E}(\eta_c(t)) \beta_{\omega}^T(t)\tilde{B}_{1i}^Tdt+\tilde{B}_{1i}\beta_{\omega}(t){\rm E}(\eta_c^T(t)) dt\\
&+\tilde{B}_{1i} S_{\tilde{\omega}}(t)\tilde{B}_{1i}^Tdt+\tilde{B}_{2i} S_{\nu_K}(t)\tilde{B}_{2i}^Tdt.
\end{aligned}
\end{equation}

We further calculate
\begin{equation}
\label{eq53}
\begin{aligned}
\frac{d\langle \eta(t) \rangle}{dt}&=\tilde{A}_i\langle \eta(t) \rangle+\tilde{B}_{1i}\beta_{\omega}(t),\\
\frac{{\rm E}(\eta_c(t))}{dt}&=\tilde{A}_i{\rm E}(\eta_c(t))+\tilde{B}_{1i}\beta_{\omega}(t).
\end{aligned}
\end{equation}
Note that if we let the mean of the Gaussian state $ \langle \eta(0) \rangle=\check{\eta}_0=\check{\eta}_{c0}={\rm E}(\eta_c(0))$, we then have $\langle \eta(t) \rangle \equiv{\rm E}(\eta_c(t))$. Moreover, we obtain $Q(t)\equiv Q'(t)$.

For the classical system \eqref{eq47}, we define the storage function as $\langle V(\eta_c(t)) \rangle={\rm E}(\eta_c^T(t)P(t)\eta_c(t))$, and the supply function as
\begin{equation}
\label{eq51_1}
\begin{aligned}
&\gamma(\eta_c(t), \beta_{\omega}(t))\\
&=\beta_{zc}^T(t)\beta_{zc}(t)-g^2 \beta_{\omega}^T(t)\beta_{\omega}(t)+\epsilon \eta_c^T(t)\eta_c(t)+\epsilon \beta_\omega^T(t)\beta_\omega(t)\\
&=\eta_c^T(t)\left[\tilde{C}_i^T\tilde{C}_i+\epsilon \mbox{I}\right]\eta_c(t)-(g^2-\epsilon) \beta_{\omega}^T(t)\beta_{\omega}(t),
\end{aligned}
\end{equation}
and we have $${\rm E}(\eta_c^T(t)P(t)\eta_c(t))={\rm Tr}\left[ \tilde{P}(t)Q'(t) \right],$$
$${\rm E}\left\{\eta_c^T(t)\left[\tilde{C}_i^T\tilde{C}_i+\epsilon \mbox{I}\right]\eta_c(t)\right\}={\rm Tr}\left\{\left[\tilde{C}_i^T \tilde{C}_i+\epsilon \mbox{I}\right]Q'(t)\right\}.$$

Since $Q(t)=Q'(t)$, \eqref{eq50} holds if the following equation holds
\begin{equation}
\label{eq51}
\begin{aligned}
&{\rm Tr}\langle P(t)Q'(t) \rangle +\int_0^t {\rm Tr}\left\{\left[\tilde{C}_i^T \tilde{C}_i+\epsilon \mbox{I}\right]Q'(t)\right\}ds\\
&-(g^2-\epsilon) \int_0^t \beta_{\omega}^T(s)\beta_{\omega}(s)ds \\
&\leq \langle \eta_c^T(0) P_0 \eta_c(0) \rangle+\lambda t.
\end{aligned}
\end{equation}

Eq. \eqref{eq51} indicates that the classical system \eqref{eq43} is strictly bounded real with disturbance attenuation $g$. Hence, we conclude that if the classical system with the controller in the form of \eqref{eq44} is strictly bounded real with disturbance attenuation $g$, the quantum system with the same control parameters in controller \eqref{eq32} is also strictly bounded real with $g$. The following proposition on $H^\infty$ control design for classical systems,  is cited here and it will be applied to the quantum case in this paper.
\begin{proposition}\cite{de2000output}
\label{propos1}
If there exists $P=(P_1, \cdots, P_N), P_i >0$ satisfies
\begin{equation}
\label{eq33}
A_i^TP_i+P_iA_i+\sum_{j=1}^N \pi_{ii}P_j+g^{-2}P_iB_1B_1^TP_i+C_1^TC_1 <0,
\end{equation} for $i=1, \cdots, N$, then $\left\| T_{\omega z}\right\|_\infty <g$.
\end{proposition}
Here, the norm $\left\| T_{\omega z}\right\|_\infty $ is the $H^\infty$-norm for the system from disturbance input $\omega(t)$ to the error output $z(t)$. Now, we have the following conclusion for quantum systems under consideration.

\begin{theorem}
\label{the2}
If there exists a controller of the form \eqref{eq32} such that the closed-loop system \eqref{eq35} is strictly bounded real with disturbance attenuation $g$, then the linear matrix inequalities (LMIs) \eqref{eq37}-\eqref{eq38} have feasible solutions $X_i, Y_i$ and $L_i, F_i$, where for $i=1, \cdots, N$, we define
\[
S_i(Y)=-{\rm diag}\left(Y_1, \cdots, Y_{i-1}, Y_{i+1}, \cdots, Y_N\right),
\]
and
\[
R_i(Y)=\left[\begin{smallmatrix}
\sqrt{\pi_{1i}}Y_i & \cdots & \sqrt{\pi_{(i-1)i}}Y_i & \sqrt{\pi_{(i+1)i}}Y_i & \cdots & \sqrt{\pi_{Ni}}Y_i
\end{smallmatrix}\right].
\]
\begin{figure*}[hb]
\hrulefill
\begin{equation}
\label{eq37}
\begin{bmatrix}
A_i^TX_i +X_iA_i+L_iC_2+C_2^TL_i^T+C_1^TC_1+\sum_{j=1}^N \pi_{ij}X_j & X_iB_1+L_iD_{2}\\
B_1X_i+D_{2}^TL_i^T & -g^2\mbox{I}
\end{bmatrix}<0.
\end{equation}
\end{figure*}
\begin{figure*}[hb]
\begin{equation}
\label{eq38}
\begin{bmatrix}
Y_i &\mbox{I}\\
\mbox{I} & X_i
\end{bmatrix}>0,
\begin{bmatrix}
A_iY_i +Y_iA_i^T+B_2F_i+F_i^TB_2^T+\pi_{ii}Y_i+g^{-2}B_1B_1^T & (C_1Y_i+D_{1}F_i)^T & R_i(Y)\\
C_1Y_i+D_{1}F_i & -\mbox{I}& 0\\
R_i^T(Y) & 0 & S_i(Y)
\end{bmatrix}<0.
\end{equation}
\end{figure*}
In this case, the controller is given by
\begin{equation}
\label{eq39}
\mathcal{C}_i=F_iY_i^{-1},
\end{equation}
\begin{equation}
\label{eq40}
\mathcal{B}_i=(Y_i^{-1}-X_i)^{-1}L_i,
\end{equation}
\begin{equation}
\label{eq41}
\mathcal{A}_i=(Y_i^{-1}-X_i)^{-1}M_iY_i^{-1},
\end{equation}
where $M_i=-A_i^T-X_iA_iY_i-X_iB_2F_i-L_iC_2Y_i-C_1^T(C_1Y_i+D_{12}F_i)-g^{-2}(X_iB_1+L_iD_{21})B_1^T-\sum_{j=1}^N \pi_{ij}Y_j^{-1}Y_i$.

Similarly, if the LMIs \eqref{eq37} and \eqref{eq38} have feasible solutions and the controller is defined as in \eqref{eq39}-\eqref{eq41}, then the closed-loop system \eqref{eq35} is strictly bounded real with the disturbance attenuation $g$.
\end{theorem}
\begin{proof}
According to Theorem \ref{the1}, this theorem can be proved in a straightforward way using the corresponding classical $H^\infty$ control results in \cite{de2000output}.
\end{proof}

\subsection{Physical realization of the controller}
\label{physicalrealisation}
Unlike the classical systems in which we assume a differential equation always corresponds to some real plant, the linear differential equation in \eqref{eq31} does not necessarily represent a meaningful quantum system. Sufficient and necessary conditions for physical realization of the following time-invariant linear system have been proposed in \cite{R1}:
\begin{equation}
\label{eq1}
\begin{aligned}
dx(t)&=Ax(t)dt+Bd\omega(t), x(0)=x_0;\\
dy(t)&=Cx(t)dt+Dd\omega(t),
\end{aligned}
\end{equation}
where the first condition is to preserve the commutation relation (CR) during the whole evolution time. In this paper, we need to consider the physical realization of time-varying linear quantum systems.

The following proposition shows sufficient and necessary conditions for preserving CR of systems \eqref{eq8} with only the parameter $A(t)$ being time-varying.
\begin{proposition}
\label{propos2}
The time-invariant system \eqref{eq1} preserves CR, if and only if the system \eqref{eq8} preserves CR with time evolving.
\end{proposition}
\begin{proof}
According to Theorem 2.1 in \cite{R1} (see Appendix for details), if the system \eqref{eq1} preserves CR, one has
\begin{equation}
\label{eq10}
\begin{aligned}
&\mbox{i}A\Theta+\mbox{i}\Theta A^T+BT_{\tilde{\omega}}B^T\\
&=\mbox{i} 2\Theta [R+\Im (\Lambda^{\dagger} \Lambda)] \Theta+\mbox{i} \Theta \left(2\Theta[R+\Im (\Lambda^{\dagger} \Lambda)] \right)^T+B T_{\tilde{\omega}} B^T\\
&=2\mbox{i}\Theta \Im(\Lambda^{\dagger} \Lambda)\Theta+2\mbox{i}\Theta[\Im (\Lambda^{\dagger} \Lambda)]^T
+B T_{\tilde{\omega}} B^T\\
&=0.
\end{aligned}
\end{equation}

By replacing $A$ with $A(t)=2\Theta (R(F(t))+\Im(\Lambda^{\dagger})\Lambda)$, we have
\begin{equation}
\label{eq12}
\begin{aligned}
&\mbox{i} 2\Theta [R(t)+\Im (\Lambda^{\dagger} \Lambda)] \Theta+\mbox{i} \Theta \left(2\Theta[R(t)+\Im (\Lambda^{\dagger} \Lambda)] \right)^T+B T_{\tilde{\omega}} B^T\\
&=2\mbox{i}\Theta \Im(\Lambda^{\dagger} \Lambda)\Theta+2\mbox{i}\Theta[\Im (\Lambda^{\dagger} \Lambda)]^T
+B T_{\tilde{\omega}} B^T\\
&=0.
\end{aligned}
\end{equation}
That is,
\begin{equation}
\label{eq9}
\mbox{i}A(t)\Theta+\mbox{i}\Theta A^T(t)+BT_{\tilde{\omega}}B^T=0.
\end{equation}
Hence, the time-varying parameter $A(t)$ does not affect the CR of quantum systems, which proves the proposition.
\end{proof}

Based on this proposition, we find that the sufficient and  necessary conditions of physical realizability can be directly extended to the following proposition from Theorem 3.4 in \cite{R1}.

\begin{proposition}
\label{propos3}
The system \eqref{eq8} is physically realizable if and only if
\begin{equation}
\label{eq21-a}
\mbox{i}A(t)\Theta+\mbox{i}\Theta A^T(t)+B T_{\omega}B^T =0,
\end{equation}
\begin{equation}
\label{eq21}
\begin{aligned}
B\begin{bmatrix}
\mbox{I}_{n_y \times n_y}\\
o_{(n_{\omega}-n_y)\times n_y}
\end{bmatrix}&=\Theta C^TP_{N_y}^T \\
&\times \begin{bmatrix}
0_{N_y \times N_y} & \mbox{I}_{N_y \times N_y} \\
-\mbox{I}_{N_y \times N_y} & 0_{N_y \times N_y}
\end{bmatrix}P_{N_y} \\
&=\Theta C^T {\rm diag}_{N_y}(J).
\end{aligned}
\end{equation}
\end{proposition}
\begin{proof}
This proposition can be proved directly by  using Proposition \ref{propos2} and Theorem 3.4 in \cite{R1}.
\end{proof}

Furthermore, for a more general time-varying quantum system described by the following equation
\begin{equation}
\label{eq22}
\begin{aligned}
dx(t)&=A(t)x(t)dt+B(t)d\omega(t),\\
dy(t)&=C(t)x(t)dt+D(t)d\omega(t),
\end{aligned}
\end{equation}
we have the following corollary.
\begin{corollary}
\label{coro1}
The system \eqref{eq22} is physically realizable if and only if the following two equations
\begin{equation}
\label{eq23-a}
\mbox{i}A(t)\Theta+\mbox{i}\Theta A^T(t)+B(t) T_{\omega}B^T(t) =0,
\end{equation}
\begin{equation}
\label{eq23-b}
\begin{aligned}
B(t)\begin{bmatrix}
\mbox{I}_{n_y \times n_y}\\
o_{(n_{\omega}-n_y)\times n_y}
\end{bmatrix}&=\Theta C^T(t)P_{N_y}^T \\
&\times \begin{bmatrix}
0_{N_y \times N_y} & \mbox{I}_{N_y \times N_y} \\
-\mbox{I}_{N_y \times N_y} & 0_{N_y \times N_y}
\end{bmatrix}P_{N_y} \\
&=\Theta C^T(t) {\rm diag}_{N_y}(J),
\end{aligned}
\end{equation}
hold for each time $t$.
\end{corollary}
\begin{proof}
The corollary can be proved directly from the proof of Theorem 3.4 in \cite{R1} at each time $t$ and the proof is omitted.
\end{proof}

With Proposition \ref{propos3} and Corollary \ref{coro1}, it is not difficult to check the physical realizability of the controller $\{\mathcal{A}_i, \mathcal{B}_i, \mathcal{C}_i\}$ designed by the $H^\infty$ control method. It should be noted that Theorem \ref{the2} only gives the parameters $\{\mathcal{A}_i, \mathcal{B}_i, \mathcal{C}_i\}$, while the parameters $\mathcal{E}_i$ and $\mathcal{D}_i$ are not constructed from the $H^\infty$ control design method, which gives the flexibility to obtain a physically realizable controller by introducing additional quantum noises and constructing $\mathcal{E}_i$ and $\mathcal{D}_i$. Theorem 5.5 in \cite{R1} illustrates that for any given $\mathcal{A}_i, \mathcal{B}_i, \mathcal{C}_i$, there exist controller parameters $\mathcal{E}_i, \mathcal{D}_i$ and quantum noises $\nu_K$ such that the controller \eqref{eq32} is physically realizable. The algorithm to construct a realizable controller can be found in \cite{vuglar2016quantum}, which can be used here for each controller $i$.

\section{Applications of fault-tolerant control to quantum optical systems}
\label{application}
In this section, we consider possible applications of fault-tolerant coherent $H^\infty$ control of linear quantum systems in quantum optics. Here, we consider that a quantum plant suffers from a fault signal and a coherent feedback controller is designed to make the quantum system fault tolerant and robust against external disturbance inputs. We first give a brief introduction to some necessary optical components that will be used to set up the quantum plant as well as the controller.
\subsection{Basic optical components}
\subsubsection{Cavities}
The optical cavity is a widely used component in quantum optics. It is composed of two mirrors in the simplest case as well as a group of mirrors for more general cases. The following linear differential equation describes the dynamics of an optical field in an empty cavity with $m$ mirrors,
\begin{equation}
\label{eq61}
d\hat{a}(t)=-\frac{\kappa}{2}\hat{a}(t)dt+\sum_{j=1}^{m}\sqrt{\kappa_j}d\hat{A}_{in,j}(t).
\end{equation}
Here, $\hat{a}(t)$ is the annihilation operator inside the cavity, $\kappa=\sum_{j=1}^m \kappa_j$ represents the total decay rates while $\kappa_j$ is the decay rate for the $j$-th mirror. Also $\hat{A}_{in,j}$ represents the input field of the $j$-th mirror. The output of each mirror is described as
\begin{equation}
\label{eq62}
d\hat{A}_{out,j}(t)=\sqrt{\kappa_j}\hat{a}(t)dt-d\hat{A}_{in,j}(t).
\end{equation}

\subsubsection{Second-order nonlinear effect}
Nonlinear crystals are materials causing nonlinear effects and have been widely used in many quantum optical experiments. Squeezed states are generated via a second-order nonlinear process where a fundamental field (oscillating at $2\nu$) is coupled with the second harmonic field (oscillating at $2\nu$). Its Hamiltonian is given by \cite{bachor2004guide}
\begin{equation}
\label{eq54}
H=\mbox{i}\hbar \chi^{(2)}(\hat{b}^{\dagger}\hat{a}^2-(\hat{a}^{\dagger})^2\hat{b}),
\end{equation}
where $\hat{a}$ is the annihilation operator of the fundamental field while $\hat{b}$ is that of the harmonic field. We may understand this Hamiltonian from the up-conversion and down-conversion aspect, where the first term can be explained as the annihilation of two photons at the fundamental frequency and the creation of one at the harmonic frequency, while the second term represents the reverse process. Hence, the Heisenberg equations of these two fields are written as
\begin{equation}
\label{eq55}
\begin{aligned}
\dot{\hat{a}}&=-2\chi^{(2)}\hat{a}^{\dagger}\hat{b},\\
\dot{\hat{b}}&=\chi^{(2)}\hat{a}^2,
\end{aligned}
\end{equation}
where $\chi^{(2)}$ is the second-order nonlinear coefficient of the crystal.

\subsubsection{Dynamical squeezers}
An optical parametric oscillator (OPO) is composed of several mirrors and a nonlinear crystal. The nonlinear interaction then can be enhanced by a cavity. By combining the equations of motion from the crystal and an empty cavity, we obtain the dynamical equation of the internal cavity mode as
\begin{equation}
\label{eq56}
d\hat{a}(t)=-\frac{\kappa}{2} \hat{a}(t)dt-2\chi^{(2)}\hat{a}^{\dagger}(t)\hat{b}(t)dt+\sum_{j=1}^m\sqrt{\kappa_j}d\hat{A}_{in,j}(t),
\end{equation}
while the output equations remain as in \eqref{eq62}.

We assume the harmonic field (or what we call the pumping field) is an intense field, which can be undepleted by its interaction with the nonlinear crystal and the fundamental field of interest \cite{bachor2004guide}. Under this assumption, we can replace the operator $\hat{b}$ with a complex number $\beta$, which means we may ignore the dynamics of the pumping field, leaving the dynamics of fundamental mode as
\begin{equation}
\label{eq63}
d\hat{a}(t)=-\frac{\kappa}{2} \hat{a}(t)dt-\chi\hat{a}^{\dagger}(t)dt+\sum_{j=1}^m\sqrt{\kappa_j}d\hat{A}_{in,j}(t),
\end{equation}
where $\chi=2\chi^{(2)}\beta$. Here, we only consider the case $\chi\in \mathbb{R}$. The parameter $\beta=|\beta|e^{i\phi}$, is composed of the real amplitude $|\beta|$ and the phase $\phi$ of the pump field. We assume only the amplitude of the pump field suffers from the fault signal, while its phase remains unchanged as $\phi=2k\pi, k=1,2,\cdots$. Under this assumption, $\chi$ becomes real. This parameter is pumping dependent, which means we can change its value by changing the pumping field. In this paper, we define this class of OPO as a dynamical squeezer, which has been widely used in quantum optical experiments to generate the squeezing states.

\subsubsection{Static squeezers}
For some practical applications, we may not be interested in the internal dynamics of a squeezer and only focus on the transformation matrix between its input and output fields. To obtain this relation, the evolution time inside the cavity may be assumed to be very short. This relation has been obtained in \cite{petersen2011realization} by assuming the evolution time of the fundamental mode is extremely short. To directly understand the static squeezer from an experimental point of view, we first transform the operators of cavity mode, input and output fields to the frequency domain by using the Fourier transform:
\begin{equation}
\label{eq57}
\begin{aligned}
-\mbox{i}\omega\hat{a}(\omega)&=-\frac{\kappa}{2}\hat{a}(\omega)-\chi\hat{a}^{\dagger}(-\omega)-i\omega\sum_{j=1}^m \sqrt{\kappa_j}\hat{A}_{in,j}(\omega),\\
-\mbox{i}\omega \hat{a}^{\dagger}(-\omega)&=-\frac{\kappa}{2} \hat{a}^{\dagger}(-\omega)-\chi\hat{a}(\omega)-i\omega\sum_{j=1}^m\sqrt{\kappa_j}\hat{A}_{in,j}^{\dagger}(-\omega),
\end{aligned}
\end{equation}

\begin{equation}
\label{eq64}
(-i\omega)\hat{A}_{out,j}(\omega)=\sqrt{\kappa_j}\hat{a}(\omega)-(-i\omega)\hat{A}_{in,j}(\omega).
\end{equation}
By solving the two equations in \eqref{eq57}, we obtain
\begin{equation}
\label{eq58}
\begin{aligned}
\hat{a}(\omega)&=\frac{-i\omega}{(\frac{\kappa}{2}-\mbox{i}\omega)^2-|\chi|^2}\sum_{j=1}^m\bigg\{ \left[(\frac{\kappa}{2}-\mbox{i}\omega)\sqrt{\kappa_j}\right]\hat{A}_{in,j}(\omega) \\
&-\chi \sqrt{\kappa_j} \hat{A}_{in,j}^{\dagger}(-\omega)\bigg\}.
\end{aligned}
\end{equation}
Substituting \eqref{eq58} into \eqref{eq64}, we obtain the transfer function for the squeezer as
\begin{equation}
\label{eq59}
\begin{aligned}
&\hat{A}_{out,j}(\omega)=\frac{1}{(\frac{\kappa}{2}-\mbox{i}\omega)^2-|\chi|^2}\cdot\\
&\bigg\{\sum_{j=1}^k \left[(\frac{\kappa}{2}-\mbox{i}\omega)\kappa_j\hat{A}_{in,j}(\omega)-\chi \kappa_j \hat{A}_{in,j}^{\dagger}(-\omega)\right]\\
&-\left[(\frac{\kappa}{2}-\mbox{i}\omega)^2-|\chi|^2\right]\hat{A}_{in,j}(\omega)\bigg\}.
\end{aligned}
\end{equation}

The assumption that the evolution time is extremely short means that the squeezer has a broad squeezing spectrum. Squeezers with broad bandwidth have been realized in many experimental setups such as a monolithic cavity \cite{yonezawa2010generation} and a single-pass waveguide \cite{yoshino2007generation}. For this kind of squeezer, if we only focus on a range of frequencies, the noise power can be taken as a constant, which means the relation between input and output of the squeezer remains unchanged with frequency. Without loss of generality, we consider the case with $\omega \ll \kappa$. We can then obtain a simplification of \eqref{eq59}
\begin{equation}
\label{eq65}
\begin{aligned}
\hat{A}_{out,j}&=\frac{1}{(\frac{\kappa}{2})^2-\chi^2}\cdot\bigg\{\sum_{j=1}^m \left[\frac{\kappa}{2}\kappa_j\hat{A}_{in,j}-\chi \kappa_j \hat{A}_{in,j}^{\dagger}\right]\\
&-\left[\frac{\kappa}{2}^2-\chi^2\right]\hat{A}_{in,j}\bigg\}.
\end{aligned}
\end{equation}

A squeezer with the relationship between input and output as in \eqref{eq65} is called static squeezer, which has been proposed in \cite{petersen2011realization}. From the perspective of quantum control theory, we usually analyze a system in the time domain. In this case, the static squeezer can be seen as a squeezer at a stable state, where the dynamics between input and output have been assumed to be static. The relation has been deduced in \cite{petersen2011realization}, where the time of light going through the squeezer has been assumed to be short enough. While from the experimental point of view, the dynamics are usually analyzed in the frequency domain by using the Fourier transformation. In this case, the static squeezer is an approximation of a dynamical squeezer with a broad squeezing spectrum, and the output of the fundamental field at any frequency within the interested frequency range remains constant. This assumption results in \eqref{eq65} for the relation between input and output.

\subsection{Fault-tolerant control design for quantum optical systems}
In this section, we consider a linear quantum system arising in quantum optics. When this system suffers from a fault signal an $H^\infty$ coherent feedback controller can be designed to deal with the fault process as well as the disturbance input. The system has been designed to generate squeezed light in \cite{serikawa2016creation}. The system is a dynamical squeezer composed of three mirrors, and its simplified diagram is shown as in Fig. \ref{system}.

\begin{figure}
\centering
\includegraphics[width=3.5in]{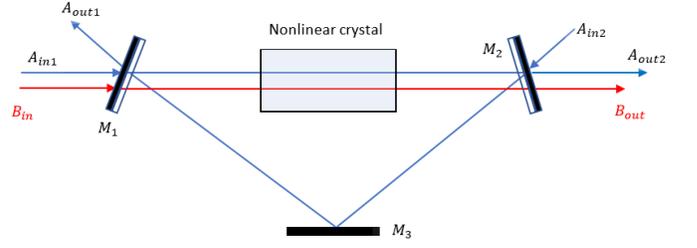}
\caption{Schematic of the OPO composed of three mirrors. $M_1$ and $M_2$ have partial transmissivity for the fundamental field and high transimissivity for the pumping field; $M_3$ is a fully reflective mirror for the fundamental field.}
\label{system}
\end{figure}

We obtain the differential equations of motion for this dynamical squeezer as
\begin{equation}
\label{eq24}
\begin{aligned}
d\hat{a}(t)&=\left[-\frac{\kappa}{2} \hat{a}(t)- \chi\hat{a}^{\dagger}(t)\right]dt\\
&+\sqrt{\kappa_{1}}d\hat{A}_{in1}(t)+\sqrt{\kappa_{2}} d\hat{A}_{in2}(t), \\
d\hat{a}^{\dagger}(t)&=\left[-\frac{\kappa}{2} \hat{a}^{\dagger}(t)-\chi\hat{a}\right]dt\\
&+\sqrt{\kappa_{1}}d\hat{A}_{in1}^{\dagger}(t)+\sqrt{\kappa_{2}}d\hat{A}_{in2}^{\dagger}(t),
\end{aligned}
\end{equation}
\begin{equation}
\label{eq25}
\begin{aligned}
&d\hat{A}_{out1}(t)=\sqrt{\kappa_{1}}\hat{a}(t)dt-d\hat{A}_{in1}(t),\\
&d\hat{A}_{out2}(t)=\sqrt{\kappa_{2}}\hat{a}(t)dt-d\hat{A}_{in2}(t),
\end{aligned}
\end{equation}
where $\kappa_1$ and $\kappa_2$ are decay rates for the mirror $M_1$ and mirror $M_2$, and $\kappa=\kappa_1+\kappa_2$. The system has two input fields $\hat{A}_{in1}$, $\hat{A}_{in2}$, and two output fields $\hat{A}_{out1}$ and $\hat{A}_{out2}$. $B_{in}$ and $B_{out}$ are the input and output of the pumping field.

To ensure that the operators are self-adjoint and work with real-valued coefficients, we write the amplitude and phase quadrature as $$x=\begin{bmatrix}
\hat{a}(t)+\hat{a}^{\dagger}(t)\\
-\mbox{i}(\hat{a}(t)-\hat{a}^{\dagger}(t))
\end{bmatrix},$$
and denote
\[
\begin{aligned}
\omega(t)&=\begin{bmatrix}
\hat{A}_{in1}(t)+\hat{A}_{in1}^{\dagger}(t)\\
-\mbox{i}(\hat{A}_{in1}(t)-\hat{A}_{in1}^{\dagger}(t))
\end{bmatrix},\\
u(t)&=\begin{bmatrix}
\hat{A}_{in2}(t)+\hat{A}_{in2}^{\dagger}(t)\\
-\mbox{i}(\hat{A}_{in2}(t)-\hat{A}_{in2}^{\dagger}(t))
\end{bmatrix},\\
z(t)&=\begin{bmatrix}
\hat{A}_{out2}(t)+\hat{A}_{out2}^{\dagger}(t)\\
-\mbox{i}(\hat{A}_{out2}(t)-\hat{A}_{out2}^{\dagger}(t))
\end{bmatrix},\\
y(t)&=\begin{bmatrix}
\hat{A}_{out1}(t)+\hat{A}_{out1}^{\dagger}(t)\\
-\mbox{i}(\hat{A}_{out1}(t)-\hat{A}_{out1}^{\dagger}(t))
\end{bmatrix}.
\end{aligned}
\]
The system may be described by
\begin{equation}
\label{eq26}
\begin{aligned}
dx(t)&=A(t)x(t)dt+B_1d\omega(t)+B_2du(t),\\
dz(t)&=C_1 x(t)dt+D_1du(t),\\
dy(t)&=C_2 x(t)dt+D_2d\omega(t),
\end{aligned}
\end{equation}
where
\[
\begin{aligned}
& A(t)=\begin{bmatrix}
-\frac{\kappa}{2}-\chi(t) & 0\\
0 & \chi(t)-\frac{\kappa}{2}
\end{bmatrix},\\
& B_1=\sqrt{\kappa_{1}}\mbox{I}, B_2=\sqrt{\kappa_{2}}\mbox{I},\\
& C_1=\sqrt{\kappa_{2}}\mbox{I}, D_1=-\mbox{I},\\
& C_2=\sqrt{\kappa_1}\mbox{I}, D_2=-\mbox{I}.
\end{aligned}
\]
Here $\mbox{I}$ is the $2\times 2$ identity matrix, and we write $\chi(t)$ as a time-varying parameter since this value may change with the fault process.

Since the pump laser is treated in a classical way, if the macroscopic laser device is subject to an undesired fault signal, a time-varying Hamiltonian is introduced.  In some practical applications, we may assume that the amplitude of the laser is not changing with time continuously and only jumps among several values. This makes it reasonable to model the fault process as a Markovian chain. Therefore, the whole system is a Markovian jump linear quantum system. To deal with this fault process, as well as the disturbance input that the dynamical squeezer itself suffers from, a coherent feedback controller is designed and connected to the plant directly without any measurement. After applying a controller to the plant, the whole closed system is shown in Fig. \ref{systemandcontrol}.
\begin{figure}
\centering
\includegraphics[width=3.5in]{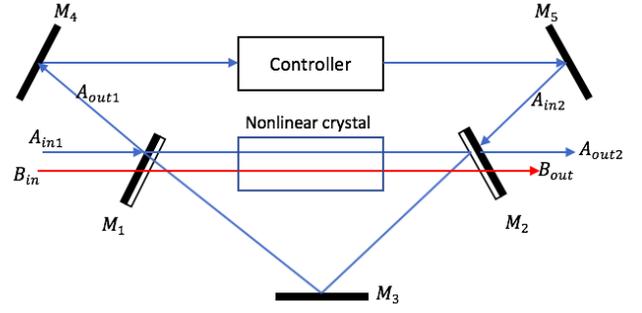}
\caption{OPO with a controller.}
\label{systemandcontrol}
\end{figure}

For the experimental system in \cite{serikawa2016creation}, the round-trip length is $45$ mm and the corresponding optical path length is $53$ mm (including the length of crystal). One of the mirrors with a partial-reflection coating has power transmissivity $T_1=14.6\% $ at $860$ nm. The other two mirrors have power transmissivity $T_2=0.02\%$ and $T_3=0$, respectively. We then calculate the decay rates by using $\kappa_i=\frac{T_i}{\tau}$ with $\tau$ being the cavity round trip time as $\kappa_{1}=0.8264$, $\kappa_{2}=0.0011$, and $\kappa=0.8275$. We here only consider the case that the system is acting as an amplifier. In this case $\chi\leq \kappa$. In the numerical example, we take three different values $\chi \in \{0.1\kappa, 0.2\kappa, 0.3\kappa\}$, which results in three modes for the system with
\begin{equation}
\label{eq27}
\begin{aligned}
A_1&=\begin{bmatrix}
-0.4551 & 0\\
0 & -0.3724
\end{bmatrix},\\ A_2&=\begin{bmatrix}
-0.4965 & 0\\
0 & -0.3310
\end{bmatrix}, \\A_3&=\begin{bmatrix}
-0.5379 & 0\\
0 &-0.2896
\end{bmatrix}.
\end{aligned}
\end{equation}
We first consider the case where the transition rate matrix is known as
$$\begin{bmatrix}
-0.02 & 0.01 & 0.01\\
0.01 & -0.01 & 0\\
0.01 & 0 & -0.01
\end{bmatrix}.$$

By solving the LMIs in \eqref{eq37} and \eqref{eq38}, we obtain the controller as
\begin{equation}
\label{eq28}
\begin{aligned}
\mathcal{A}_{1}&=\begin{bmatrix}
-1.7535 & 0\\
0 & -2.1226
\end{bmatrix} ,
\mathcal{B}_{1}=\begin{bmatrix}
1.2524 & 0\\
0 & 1.8944
\end{bmatrix},\\
\mathcal{C}_{1}&=\begin{bmatrix}
-0.0331& 0\\
0 & -0.0331
\end{bmatrix};
\end{aligned}
\end{equation}

\begin{equation}
\label{eq29}
\begin{aligned}
\mathcal{A}_{2}&=\begin{bmatrix}
-1.5796 & 0\\
0 & -2.2738
\end{bmatrix} ,
\mathcal{B}_{2}=\begin{bmatrix}
0.9713 & 0\\
0 & 2.2099
\end{bmatrix},\\
\mathcal{C}_{2}&=\begin{bmatrix}
-0.0331 & 0\\
0 & -0.0331
\end{bmatrix};
\end{aligned}
\end{equation}

\begin{equation}
\label{eq30}
\begin{aligned}
\mathcal{A}_{3}&=\begin{bmatrix}
-1.3992 & 0\\
0 & -2.4340
\end{bmatrix} ,
\mathcal{B}_{3}=\begin{bmatrix}
0.7024 & 0\\
0 & 2.5600
\end{bmatrix},\\
\mathcal{C}_{3}&=\begin{bmatrix}
-0.0331 & 0\\
0 & -0.0331
\end{bmatrix}.
\end{aligned}
\end{equation}
Here, $\{\mathcal{A}_i, \mathcal{B}_i, \mathcal{C}_i\}$ are the parameters of the $i$-th mode of the controller.

By using Proposition \ref{propos3}, it can be checked that the controller for each $i=1, 2, 3$ is not physically realizable without additional quantum noises. Hence, we use the algorithm in \cite{vuglar2016quantum} to construct the system parameters corresponding to the additional quantum noiese as $\tilde{\mathcal{B}}_{1}=\begin{bmatrix}\mathcal{B}_{1} & \mathcal{E}_{11} & \mathcal{E}_{12} \end{bmatrix}$, $\tilde{\mathcal{B}}_{2}=\begin{bmatrix}\mathcal{B}_{2} & \mathcal{E}_{21} & \mathcal{E}_{22} \end{bmatrix}$, and $\tilde{\mathcal{B}}_{3}=\begin{bmatrix}\mathcal{B}_{3} & \mathcal{E}_{31} & \mathcal{E}_{32} \end{bmatrix}$, where:
\begin{equation}
\label{eq11}
\begin{aligned}
\mathcal{E}_{11}&=\begin{bmatrix}
0.0331 & 0\\
0 & 0.0331
\end{bmatrix}, \mathcal{E}_{12}=\begin{bmatrix}
1.2258 & 0\\
0& 1.2258
\end{bmatrix},\\
\mathcal{E}_{21}&=\begin{bmatrix}
0.0331 & 0\\
0& 0.0331
\end{bmatrix}, \mathcal{E}_{22}=\begin{bmatrix}
1.3057 & 0\\
0 & 1.3057
\end{bmatrix},\\
\mathcal{E}_{31}&=\begin{bmatrix}
0.0331 &0\\
0& 0.0331
\end{bmatrix}, \mathcal{E}_{32}=\begin{bmatrix}
1.4262 & 0\\
0 & 1.4262
\end{bmatrix}.
\end{aligned}
\end{equation}
Here, these control parameters imply that the controller has three inputs, $y, \nu_1, \nu_2$, the corresponding input matrices are $\mathcal{B}_i$, $\mathcal{E}_{i1}$, and $\mathcal{E}_{i2}$, respectively.

We further present an implementation of the controller, which should switch between different modes according to the plant. The structure diagram of the controller is shown in Fig. \ref{controller}, where the left cavity with pump field $\epsilon_1$ is the static squeezer with the total decay rate $\kappa^\prime=\kappa^\prime_1$, while the right one is the general dynamical squeezer with the total decay rate as $\kappa=\kappa_1+\kappa_2+\kappa_3$.
\begin{figure}
\centering
\includegraphics[width=3.5in]{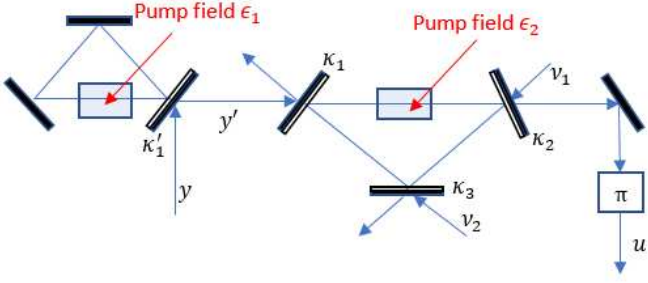}
\caption{Diagram for the controller composed of a static squeezer pumped by $\epsilon_1$, an OPO pumped by $\epsilon_2$, and a phase shifter $\pi$.}
\label{controller}
\end{figure}
According to the relation between input and output of the static squeezer in \eqref{eq65}, we first can obtain the input-output equation for the left static squeezer in Fig. \ref{controller}:
\begin{equation}
\label{eq60}
\begin{aligned}
&\begin{bmatrix}
y^{\prime}+{y^{\prime}}^{\dagger}\\
(y^{\prime}-{y^{\prime}}^{\dagger})/\mbox{i}
\end{bmatrix}\\
&=
\begin{bmatrix}
-1+\frac{\frac{\kappa^{\prime}}{2}\kappa^{\prime}-\chi \kappa^{\prime}}{(\frac{\kappa^{\prime}}{2})^2-{\chi^{\prime}}^2}& 0\\
0& -1+\frac{\frac{\kappa^{\prime}}{2}\kappa^{\prime}+\chi\kappa^{\prime}}{(\frac{\kappa^{\prime}}{2})^2-{\chi^{\prime}}^2}
\end{bmatrix}\cdot
\begin{bmatrix}
y+y^{\dagger}\\
(y-y{\dagger})/\mbox{i}
\end{bmatrix}.
\end{aligned}
\end{equation}

With \eqref{eq60}, we can write its dynamical and output equations for the system in Fig. \ref{controller} as follows
\begin{equation}
\label{eq42}
\begin{aligned}
&\left.d \begin{bmatrix}
\hat{a}(t)+\hat{a}^{\dagger}(t)\\
\left( \hat{a}(t)-\hat{a}^{\dagger}(t) \right)/\mbox{i}
\end{bmatrix}\middle/dt \right.\\
&= \begin{bmatrix}
-\frac{\kappa}{2}-\chi & 0\\
0 & -\frac{\kappa}{2}+\chi
\end{bmatrix}\begin{bmatrix}
\hat{a}(t)+\hat{a}^{\dagger}(t)\\
\left( \hat{a}(t)-\hat{a}^{\dagger}(t) \right)/\mbox{i}
\end{bmatrix}\\
&+\sqrt{\kappa_1}\begin{bmatrix}
-1+\frac{\frac{\kappa^{\prime}}{2}\kappa^{\prime}-\chi \kappa^{\prime}}{(\frac{\kappa^{\prime}}{2})^2-{\chi^{\prime}}^2}& 0\\
0& -1+\frac{\frac{\kappa^{\prime}}{2}\kappa^{\prime}+\chi\kappa^{\prime}}{(\frac{\kappa^{\prime}}{2})^2-{\chi^{\prime}}^2}
\end{bmatrix}\begin{bmatrix}
y(t)+y^{\dagger}(t)\\
(y(t)-y^{\dagger}(t))/\mbox{i}
\end{bmatrix}\\
&+\sqrt{\kappa_2}\mbox{I}\ d\begin{bmatrix}
\nu_1(t)+\nu_1^{\dagger}(t)\\
(\nu_1(t)-\nu_1^{\dagger}(t))/\mbox{i}
\end{bmatrix}+\sqrt{\kappa_3}\mbox{I}\begin{bmatrix}
\nu_2(t)+\nu_2^{\dagger}(t)\\
(\nu_2(t)-\nu_2^{\dagger}(t))/\mbox{i}
\end{bmatrix},\\
&\begin{bmatrix}
u(t)+u^{\dagger}(t)\\
\left(u(t)-u^{\dagger}(t)\right)/\mbox{i}
\end{bmatrix}=-\sqrt{\kappa_2}\mbox{I} \begin{bmatrix}
\hat{a}(t)+\hat{a}^{\dagger}(t)\\
\left(\hat{a}(t)-\hat{a}^{\dagger}(t)\right)/\mbox{i}
\end{bmatrix}\\
&+\mbox{I}\begin{bmatrix}
\nu_1(t)+\nu_1^{\dagger}(t)\\
\left(\nu_1(t)-\nu_1^{\dagger}(t)\right)/\mbox{i}
\end{bmatrix}.
\end{aligned}
\end{equation}
Here, $\hat{a}(t)$ represents the annihilation operator of the fundamental field inside the dynamical OPO; $\kappa^{\prime}=\kappa_1^{\prime}$ and $\chi^{\prime}$ are the total decay rate and the pumping dependent parameter of the static squeezer, respectively; while for the dynamical squeezer with pump filed $\epsilon_2$, $\kappa=\kappa_1+\kappa_2+\kappa_3$, and $\chi$ is pumping dependent parameter. We should note that $\kappa_1$, $\kappa_2$ and $\kappa_3$ need to be changed for different modes of the controller. This can be realized by tunable mirrors. Usually it can be achieved by making the mirror as a cavity itself and manipulating the decay rate by controlling the separation of the cavity \cite{bachor2004guide}. Also, we can change the parameters $\chi$ and $\chi^{\prime}$ by changing the pump field.

Comparing \eqref{eq28}-\eqref{eq30}, \eqref{eq11}, and \eqref{eq42}, we obtain
\begin{equation}
\label{eq66}
\begin{aligned}
\mathcal{A}_i&=\begin{bmatrix}-\frac{\kappa}{2}-\chi & 0\\
0& -\frac{\kappa}{2}+\chi\end{bmatrix},\\
\mathcal{B}_i&=\sqrt{\kappa_1}\begin{bmatrix}
-1+\frac{\frac{\kappa^{\prime}}{2}\kappa^{\prime}-\chi \kappa^{\prime}}{(\frac{\kappa^{\prime}}{2})^2-{\chi^{\prime}}^2}& 0\\
0& -1+\frac{\frac{\kappa^{\prime}}{2}\kappa^{\prime}+\chi\kappa^{\prime}}{(\frac{\kappa^{\prime}}{2})^2-{\chi^{\prime}}^2}
\end{bmatrix},\\
\mathcal{E}_{i1}&=\sqrt{\kappa_2}\mbox{I},\\
\mathcal{E}_{i2}&=\sqrt{\kappa_3}\mbox{I}.
\end{aligned}
\end{equation}
 According to the parameters for different modes, we can calculate the corresponding parameters as follows:
\begin{description}
\item [Mode 1]
$$
\begin{aligned}
\kappa&=3.8761, \chi=-0.1846, \\
\kappa_{1}&=2.3724, \kappa_{2}=0.0011, \kappa_{3}=1.5026,\\
\kappa_1^{\prime}&=10,\\
\chi^{\prime}&=0.6237;
\end{aligned}
$$
\item [Mode 2]
$$
\begin{aligned}
\kappa&=3.8534, \chi=-0.3471, \\
\kappa_{1}&=2.1475, \kappa_{2}=0.0011, \kappa_{3}=1.7046,\\
\kappa^{\prime}&=10,\\
\chi^{\prime}&=1.1953;
\end{aligned}
$$
\item [Mode 3]
$$
\begin{aligned}
\kappa&=3.8332, \chi=-0.5174, \\
\kappa_{1}&=1.7981, \kappa_{2}=0.0011, \kappa_{3}=2.0340,\\
\kappa^{\prime}&=10,\\
\chi^{\prime}&=1.7650.
\end{aligned}
$$
\end{description}

Here, we have used $\chi$ and $\kappa$ to represent the pumping coefficient and decay rate for the dynamical squeezer, and $\chi^{\prime}$ and $\kappa^{\prime}$ to represent the parameters for the static squeezer. In the static squeezer, we have assumed that only one mirror has a decay rate and the other two are ideal and fully reflective. The parameters of the squeezer also show that the total decay rate remains unchanged for the three modes of the controller and only the pumping field needs to be adjusted.

\begin{remark}
It should be noted that the controller designed above is mode-dependent, which means the controller should switch with the modes of the plant. Mode-dependent controllers are also often designed in classical systems since the mode-independent one usually does not have satisfactory performance. In classical cases, the modes of the plant are often assumed to be known to the controller, and we have also made this assumption in quantum cases here. For the dynamical squeezer in Fig. \ref{system}, it is possible to know the modes of the plant by measuring the output of the pump field, which is shown as $B_{out}$. Since the pump field is an intense laser and treated in a classical way, measuring the output does not destroy any quantum information. Furthermore, since we have assumed that the parameter $\chi$ in system equation \eqref{eq63} is real, we only consider fluctuations in the amplitude of the pump field. Therefore, we can measure changes in the amplitude of $B_{out}$ without a reference laser.
\end{remark}


%

\section{Conclusion}
\label{conclusion}
In this paper, we have extended the sufficient and necessary conditions of physical realization to the time-varying linear quantum systems described by QSDEs, which plays an important role in designing a coherent feedback controller. Then the $H^\infty$ control problem has been considered and related to Riccati differential equations. This makes it possible to design an $H^\infty$ controller by solving a group of LMIs. The physical realizability of the time-varying controller is then ensured by introducing additional noises and constructing the corresponding input matrices. For a dynamic squeezer used in quantum optical experiments, the fault signal of the pump laser has been recognized and this results in a Markovian jump linear quantum system. A coherent feedback controller has been designed to bound the effect of disturbance input on the output even when the plant suffers from a fault signal. The physical realizability of the controller has been ensured theoretically, and it also has been implemented by some basic optical components including squeezers, beamsplitters and phase shifters. This paper has assumed that the system has precisely known transit rate matrix and the designed controller is mode-dependent. In some practical applications, we may only know estimated or partial values of the transit rate matrix. This leads to the control problem for Markovian jump linear systems with partly known or uncertain rate matrices, which needs to be further considered. Furthermore, the time to measure the output of the pump field and to determine the mode of the plant sometimes may not be ignored. Therefore, future research could also include time-varying $H^\infty$ control for linear quantum systems with time delays.


%

\appendices
\section{}
Theorem 2.1 from \cite{R1}: For the system \eqref{eq1}, $[x_i(0), x_j(0)]=2\mbox{i}\Theta_{ij}$ implies $[x_i(t),x_j(t)]=2\mbox{i}\Theta_{ij}$ for all $t\geq 0$ if and only if
\begin{equation*}
\label{app1}
\mbox{i}A\Theta+\mbox{i}\Theta A^T+BT_{\tilde{\omega}}B^T=0.
\end{equation*}
\section{}
The following Lemma comes from \cite{R1}.
\begin{lemma}
\label{lemma1}
Consider a real matrix $X$ and corresponding operator valued quadratic form $x^T X x$ for the system \eqref{eq13}. The following statements are equivalent.
\begin{enumerate}
\item [i)] There exists a constant $\lambda \geq 0$ such that $\langle \rho, x^TX x \rangle\leq \lambda$ for all Gaussian states $\rho$.
\item [ii)] The matrix $X$ is negative semidefinite.
\end{enumerate}
\end{lemma}
Here $\langle \rho, \cdot \rangle$ represents the expectation with respect to the Gaussian state $\rho$.



\ifCLASSOPTIONcaptionsoff
  \newpage
\fi



\bibliographystyle{IEEEtran}
\bibliography{IEEEabrv,references}
\end{document}